\documentclass[aps,showpacs]{revtex4-1}

\usepackage{amsfonts}
\usepackage{amssymb}
\usepackage[fleqn]{amsmath}
\usepackage[mathscr]{eucal}
\usepackage{graphicx}
\usepackage{subfigure}
\usepackage{multirow}
\usepackage{color}
\usepackage{comment}
\usepackage{siunitx}
\usepackage{units}
\usepackage{textcomp}
\usepackage{xcolor,colortbl}
\usepackage{pgfplots}
\usepackage{pgfplotstable}
\usepackage{paralist}
\pgfplotsset{compat=newest}
\pgfplotsset{ legend style={font=\footnotesize}}
\newtheorem{theorem}{Theorem}
\newtheorem{corollary}[theorem]{Corollary}
\newtheorem{definition}[theorem]{Definition}
\newenvironment{proof}[1][Proof]{\textbf{#1} }{\ \rule{0.5em}{0.5em}}

\newcommand{\R}{\mathbf{R}}
\newcommand{\bv}{\mathbf{v}}
\newcommand{\bc}{\mathbf{c}}

\newcommand{\bx}{\mathbf{x}}
\newcommand{\by}{\mathbf{y}}

\newcommand{\rar}{{\rightarrow}}
\newcommand{\cD}{{\cal{D}}}

\newcommand{\beq}{\begin{equation}}
\newcommand{\eeq}{\end{equation}}
\newcommand{\pf}{Perron-Frobenius}

\DeclareMathOperator\erf{erf}
\newcommand{\argmin}{\arg\!\min}

\begin{document}

\title{Finite-time Partitions for Lagrangian Structure Identification in Gulf Stream  Eddy Transport}

\date{\today}

\author{Alexandre \surname{Fabregat Tom\`as}}
\email[Email: ]{alex.fabregat@csi.cuny.edu}
\affiliation{Department of Mathematics, City University of New York,
College of Staten Island, 2800 Victory Boulevard, Staten Island, New York, NY 10314}
\author{Igor Mezi\'c}
\affiliation{College of Engineering, University of California Santa Barbara,
Santa Barbara, CA 93106}
\author{Andrew C. Poje}
\affiliation{Department of Mathematics, City University of New York,
College of Staten Island, 2800 Victory Boulevard, Staten Island, New York, NY 10314}

\begin{abstract}

We develop a methodology to identify finite-time Lagrangian structures from data and models
using an extension of the Koopman operator-theoretic methods developed for velocity fields with simple
(periodic, quasi-periodic) time-dependence. To achieve this, the notion of the Finite Time Ergodic (FiTER) partition is
 developed and rigorously justified. In combination with a   clustering-based approach, the methodology enables identification of the temporal evolution of Lagrangian structures in a 
classic, benchmark,  oceanographic transport problem, namely the cross-stream flux induced by the interaction of a meso-scale
Gulf Stream Ring eddy with the main jet. We focus on a single mixing event driven by the interaction between an energetic cold 
core ring (a cyclone), the strong jet, and a number of smaller scale cyclones and anticyclones. The new methodology enables reconstruction of 
Lagrangian structures in three dimensions and analysis of their time-evolution.
\end{abstract}


\maketitle

\section{Introduction}

Understanding the Lagrangian transport between distinct flow features is 
particularly important in the geophysical context where rotation and stratification, coupled with the extreme
geometric aspect ratio,  produce long-lived, coherent vortices in the horizontal velocity field \cite{Provenzale99}. 
The `geometric approach' to quantifying advective exchange by calculation of Lagrangian transport boundaries,
first developed to understand the dynamics of classical Hamiltonian systems, has been both widely used and 
explicitly developed in context of meso-scale oceanographic and atmospheric flows (see,
\citep{Wiggins2005, Koshel2006, Samelson2013} for reviews
in the oceanographic context).
The need for analysis of geometric structures that organize advection is not
purely academic. Predicting rates and pathways of material transport by environmental flows 
has been an essential element in the response to several recent catastrophic events, namely, the volcanic
eruption of Eyjafjallaj{\"o}kull (2010), the Deepwater Horizon oil spill
(2010), and the nuclear disaster in Fukushima (2011). These events highlight
the importance of detecting organizing geometric structures in (near) real
time from data, either measured or generated by detailed simulation models;
consequently, such problems have become a very active intersection of
dynamical systems and fluid dynamics.

Presently, there exists a large and diverse literature of techniques for detecting, quantifying and visualizing the 
organizing structures controlling Lagrangian mixing and transport in unsteady, aperiodic flows \cite{Gildor2014}. 
Conceptually, a coherent structure in
the Lagrangian frame is a set of similar trajectories, \textit{ie} trajectories that move together in any inertial
reference frame for an extended period of time. 
Methods for unambiguously defining sets of initial conditions satisfying such criteria, either explicity or by defining the set boundaries, are readily
available for autonomous or time periodic flows and rely heavily on studying the geometry of invarient manifolds defined by asymptotic time properties.   
The theory of \emph{Lagrangian Coherent Structures}~\cite{Haller2005} (LCS)
identifies barriers that organize the transport in flows with complex
time-dependence. Initially, LCS were closely associated with computation of
Finite-Time Lyapunov Exponent fields \cite{Shadden2005}; more recently, they have
been re-formulated using a variational principle \cite{Haller2011a, Haller2012,
  Blazevski2014}, which defines them as certain geodesic lines of the local
deformation field induced by the fluid flow. This new definition allows a
finer classification of LCS, both in two- and three-dimensions, based on the
type of deformation, e.g., hyperbolic, elliptic, corresponding to different
behaviors of fluid parcels in the flow. The recent review by
Haller \cite{Haller2015} gives a detailed coverage of the current
state-of-the-art in the LCS theory.

Magnitudes of the local material deformation are typically estimated by
processing velocity gradient information and typically difficult to compute precisly in the
absence of the detailed data about the velocity field, e.g., when the system
is sampled only by sparse trajectories or the underlying velocity field is noisy or
underresolved. To avoid this, alternate measures to local deformation gradient information have been
proposed as bases for defining coherent stucture boundaries \cite{rypinaetal:2011, manchoetal:2013}. 
In sparsely-sampled planar systems, 
trajectories can be represented by space-time braids --- extremely-reduced,
symbolic representations of trajectories. The resulting approach, known as
\emph{braid dynamics} \cite{Boyland2000, Allshouse2012, Thiffeault2010} has
been successful in providing lower bounds on the amount of deformation present
in the flow in limited-data settings. The obtained bounds have been used both
in design and analysis of the material advection; unfortunately, there are
currently no extensions of braid dynamics to three-dimensional flows.

Instead of looking for barriers to transport (boundaries of coherent sets), as is the case with the LCS
theory, the theory of \emph{almost-invariant sets} attempts to directly identify a collection of
sets, fixed in initial condition space, such that the exchange of material between these set is minimized
under the action of the flow over a finite time. These sets act as routes for the material transport. 
The approach is based on the \pf{} transfer operator, which models how the flow moves
distributions of points, instead of individual trajectories. The \pf{}
operator is always infinite-dimensional and linear; the identification of
almost-invariant sets is then intimately connected with approximating its
eigenfunctions. In practice, the \pf{} operator is approximated by a finite dimensional transfer operator
produced by binning initial conditions. While the \pf{} operator has been a staple of the ergodic and
probability theory since the early 20th century, it was introduced to applied,
non-probabilistic context by Dellnitz and collaborators \cite{Dellnitz1999, Dellnitz2002} as the
basis for identification of invariant sets in time-invariant dynamical
systems. The theory has since been expanded to include detection of
almost-invariant sets of autonomous systems \cite{Froyland2003, Froyland2009},
and flows with more general time dependencies including applications to
geophysical transport \cite{Froyland2010,
Froyland2010a, froyland2015}. 

Spatial invariants of dynamical systems are directly related to infinite-time averages of
functions along Lagrangian trajectories via ergodic theory \cite{Mezic:1994}. 
In this context, evidence exists that even finite-time averages of functions can enable detection of
geometric structures important for fluid transport \cite{Poje1999}.
The utility of time-averages has been corroborated on numerical and
experimentally-realizable flows with simple time dependence\cite{Mezic1999,
  Malhotra1998, Mezic2002}.
Computations of (approximately) ergodic invariant sets using 
long-time averages of a large set of averaged function have been investigated \cite{Mezic:1994, Mezic1999, Levnajic2010, Budisic2012b} and
the connection between these approaches and LCS theory has recently been made \cite{Farazmand2015a}.

Here we concentrate on the approach established in \cite{Mezic:1994, Mezic1999, Levnajic2010, Budisic2012b} to 
extend the theory based on finite-time Lagrangian averages of the velocity field. We make an explicit connection with the Koopman operator 
theory that describes evolution of observables under dynamics of flows \cite{MezicandBanaszuk:2004,Mezic:2005,Mezic:2013}. For observables, we use 
fixed sets of basis functions combined with 
machine learning (standard high-dimensional clustering algorithms) in order to determine finite-time structures within 
which the trajectories themselves possess most similar statistics. The goal is to develop a robust procedure for identifying structures that 
(1) requires only Lagrangian position data as input (2) is based on global properties of the full set of available trajectories and 
(3) naturally extends to three dimensional velocity fields.

Specifically, we seek to elucidate the time dependence of interacting coherent mesoscale features (spatial scales $>$10km) in a highly resolved, 
highly time dependent, submesoscale-eddy-permitting model of the Gulf Stream. To date, the majority of Lagrangian coherent structure 
applications in geophysical flows have considered  meso-scale dominated velocity fields either produced by relatively coarsely resolved/idealized  models 
\cite{Miller:1997,Lipp08,Bettencourt2012,Pratt2014} or derived directly from available remote sensing data using geostrophy
\cite{Ngan1999,Ovidio:2004,Beronvera2008}. In contrast to these fields which necessarily impose a rapid fall-off of energy at small spatial scales, 
current ocean models produce distinct signatures of submesocale turbulence, with broad wavenumber and frequency spectra, 
as model resolution is increased towards $ \Delta x \lesssim$ 1km \cite{Capet2008a, Bracco2016}.
A primary goal here is to demonstrate the ability of time-averaged, cluster-based partitions to identify dominant 
three-dimensional mesoscale mixing structures in the presence of significant small scale fluctuations.
The specific problem, quantifying eddy-induced mixing between and across the Gulf Stream,
was the original motivation for many LCS studies \cite{BowRos85,Samels91,DutPal94} and has received considerable attention \cite{SonRos95,Rog99,
Yuan2002,Budyansky2009}.

The paper is organized as follows: In section \ref{sectII} we discuss the model velocity field and transport geometry. 
In section \ref{sec.observables} we discuss the Koopman operator approach to visualization of Lagrangian structures and discuss the selection of observables. In section \ref{sec.partition} we  introduce the notion of Finite-Time Ergodic  (FiTER) Partition 
and present results on identification of Lagrangian coherent structures using a clustering algorithm. In section \ref{sec.reconstruction} we discuss how the technique of FiTER partition, together with the clustering algorithm enables reconstruction of three-dimensional Lagrangian structures and analysis of their evolution in time. We conclude in section \ref{sectVI}. In the Appendix, we provide the rigorous theory of FiTER partitions and define the associated  measure-theoretic notions of almost invariant sets and maximally coherent sets.

\section{Numerical Model}
\label{sectII}
\subsection{HYCOM numerical simulation for the North Atlantic} 

The Hybrid Coordinate Ocean Model (HYCOM) \cite{bleck02, halliwell04, chassignet06} is used to 
numerically simulate the Gulf Stream region. The resulting database corresponds to a high-resolution subdomain 
with a Mercator grid 
resolution of $1/48^\circ$ ($\sim 2$\si{\kilo\meter}) spanning $81.44^\circ$ W $28.78^\circ$ N to $50^\circ$ W 
$45.72^\circ$ N (see Fig. \ref{fig.salvel1}).
The vertical spatial discretization consists of 30 hybrid (z-sigma-isopycnal) 
layers of which the top six are in $z$ coordinates while the ocean interior is discretized in potential density 
coordinates. 
The high-resolution subdomain is one-way nested within a coarse grid ($1/12^\circ$) run covering the Atlantic Ocean 
and the Mediterranean Sea in the latitudinal range $28^\circ$ S and $80^\circ$ N.
The fine grid simulation, 
initialized using the results from a 8 year long coarse grid run,
is evolved for a total span of 501 days.
After discarding the initial 134 days to minimize the effects of re-gridding, 
the complete database catalogs  flow evolution over one whole year starting on May 15 (day 1 hereafter).
The model data is identical to that described in \citep{mensa2013} to which the reader is referred for dynamical
details.
In terms of Gulf-Stream characterization, velocities and eddy kinetic energy, the numerical results are consistent 
with drifter data analysis (\cite{fratantoni01, garraffo01,lumpkin13}). 
Similarly, the sea surface height variability is the 
Gulf Stream extension is in agreement with altimeter data \cite{ducet00}. 

\begin{figure}
\centering
\includegraphics{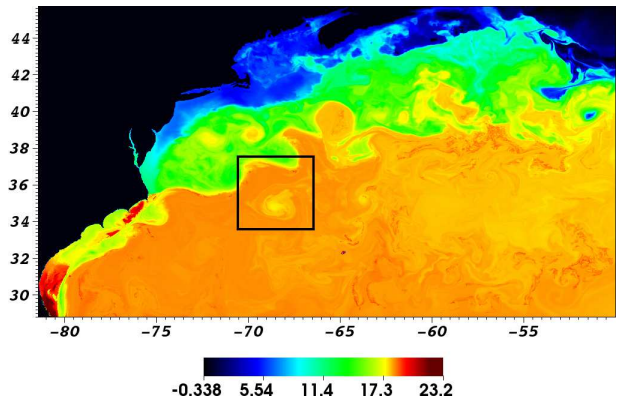}
%\begin{tikzpicture}
%\node[inner sep=5pt] (plot1) at (0,0)
%{\includegraphics[width=0.45\textwidth]{./plots/global.png}};
%\end{tikzpicture}
\caption{
North Atlantic HYCOM numerical simulation snapshot showing temperature at isopycnal 12 for day 80.
The black box shows the region of interest defined by initial extent of the Lagrangian particle patch.
}
\label{fig.salvel1}
\end{figure}

We are interested in elucidating the evolution of coherent 
structures and their associated transport where the meandering Gulf Stream generates strong 
mesoscale flow features.
Such conditions are found throughout, but we concentrate on a specific event during the summer (day 70 to day 125)
in the subregion defined by the corners $70.5^\circ W \, 33.65^\circ N$ and $66.5^\circ W \, 37.65^\circ N$
shown in Fig. \ref{fig.salvel1} as a black square.
The projection of this $[L_x \times L_y]=[4^\circ \times 4^\circ]$
domain roughtly corresponds to an extent of $[361 \times 444] \si{\kilo\meter}$.
Several HYCOM snapshots of instantaneous velocity and salinity for this temporal and spatial domain of interest
are shown in Fig. \ref{fig.salvel}.
The domain is approximately centered on a strong cyclonic, `cold-core' ring located south
of the Gulf Jet Stream characterized by strong eastward currents and low salinity.
The cold-core ring persists as a coherent feature over the 25 day time span, interacting
with the meandering jet and smaller scale anti-cyclonic structures throughout this time. 
The lower vertical boundary of the domain of interest is given by the 23rd isopycnal model layer, shown for day 80 
in Fig. \ref{fig.isopyc} along with the bathymetry, characterized by the Caryn Seamount 
rising approximately 2,000 meters above the Sohm abyssal plain.

\begin{figure}
\centering
\includegraphics{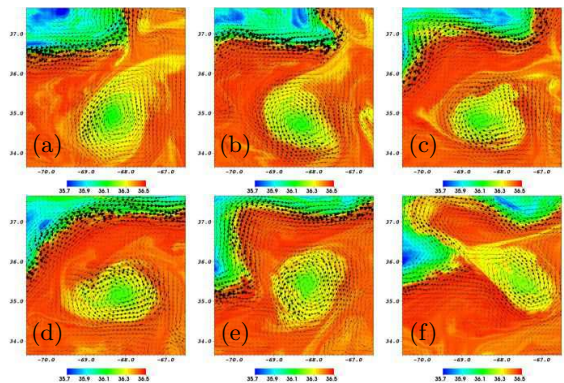}
%\begin{tikzpicture}
%\node[inner sep=0pt] (a) at (0,0){
%\includegraphics[width=0.15\textwidth]{./plots/frame0000.jpeg}
%}node at (-0.75,-0.6) {(a)};
%\node[inner sep=0pt] (b) at (2.5,0){
%\includegraphics[width=0.15\textwidth]{./plots/frame0001.jpeg}
%}node at (1.75,-0.6) {(b)};
%\node[inner sep=0pt] (c) at (5.0,0){
%\includegraphics[width=0.15\textwidth]{./plots/frame0002.jpeg}
%}node at (4.25,-0.6) {(c)};
%\node[inner sep=0pt] (d) at (0,-2.5){
%\includegraphics[width=0.15\textwidth]{./plots/frame0003.jpeg}
%}node at (-0.75,-3.1) {(d)};
%\node[inner sep=0pt] (e) at (2.5,-2.5){
%\includegraphics[width=0.15\textwidth]{./plots/frame0004.jpeg}
%}node at (1.75,-3.1) {(e)};
%\node[inner sep=0pt] (f) at (5,-2.5){
%\includegraphics[width=0.15\textwidth]{./plots/frame0005.jpeg}
%}node at (4.25,-3.1) {(f)};
%\end{tikzpicture}
\caption{
Top to bottom, left to right: instantaneous velocity vectors and salinity fields for days 70, 75, 80, 85, 90 and 95 at isopycnal 12.
}
\label{fig.salvel}
\end{figure}

\begin{figure}
\centering
\includegraphics{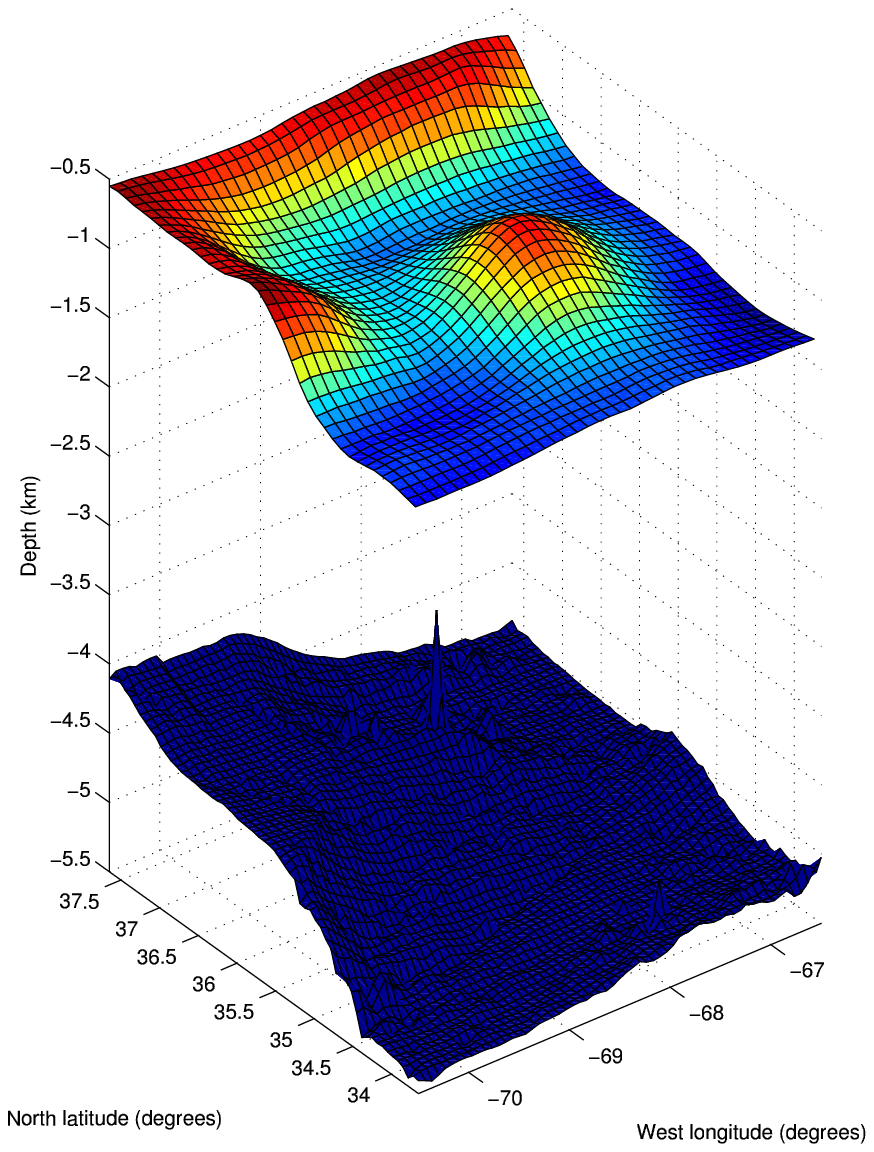}
%\begin{tikzpicture}
%\node[inner sep=0pt] (a) at (0,0){
%\includegraphics[width=0.52\textwidth]{./plots/iso_bath2.pdf}
%};
%\end{tikzpicture}
\caption{Lowest isopycnal and bathymetry extending over the initial synthetic drifters patch centered at $72.5^\circ \,
W 35.65^\circ N$.
The dominant cold core feature is readily seen.}
\label{fig.isopyc}
\end{figure}

\subsection{Synthetic trajectories}

The Lagrangian data is given by trajectories of 90,000 synthetic drifters, initialized at various times on a uniform
300 $\times$ 300 horizontal grid spanning the domain of interest in each of the 23 isopycnal layers considered.
The trajectories are computed by integrating
$d\vec{x}_p/dt = \vec{u}_p \left( \vec{x}_p,t\right)$
for each isopycnal level $i$ over the time interval $[ t_0,t_0 +\tau_i]$.
The integration scheme used is fourth order Runge-Kutta and each particle $p$ velocity
$\vec{u}_p$ is obtained via linear temporal interpolation between consecutive 12-hour HYCOM velocity fields and 
third order spline interpolation in both spatial directions.

Fig.~\ref{fig.traj1} shows the position of the full set of $90,000$ drifters on
the uppermost (black) and lowest (red) isopycnals after 7 days. The magenta lines mark the location of the initial patch.
For fixed integration time, particles in
the lowest isopycnal travel much less than those in the uppermost level as a result of the large diapycnal (vertical) 
shear of the along-isopycnal velocities.

As explained below, the approach used to organize the complexity of the full set of trajectories across 
isopycnals is based on computing time averages of sets of spatially dependent functions. 
The disparity in Lagrangian averages on different layers can be accounted for either by explicitly
choosing the layer (depth) dependence of such functions \emph{a priori}, or by considering identical sets of functions at 
each isopycnal and adjusting averaging times accordingly. For simplicity of presentation, we choose the latter 
approach and rescale integration times across isopycnals to ensure that the average length of trajectories are
equal.
The temporal rescaling used throughout is shown in Fig.~\ref{fig.traj2}. The effect on the trajectories
is illustrated in Fig.~\ref{fig.traj1} where the final positions of the 
$\tau_{23} = 4.7 \tau_{1} = 33$ day trajectories in the lowest isopycnal are 
shown in blue. Under this temporal rescaling, the average length of trajectories in the lowest (blue) and the topmost 
(black) layers are commensurate. 

\begin{figure}
\centering
\includegraphics{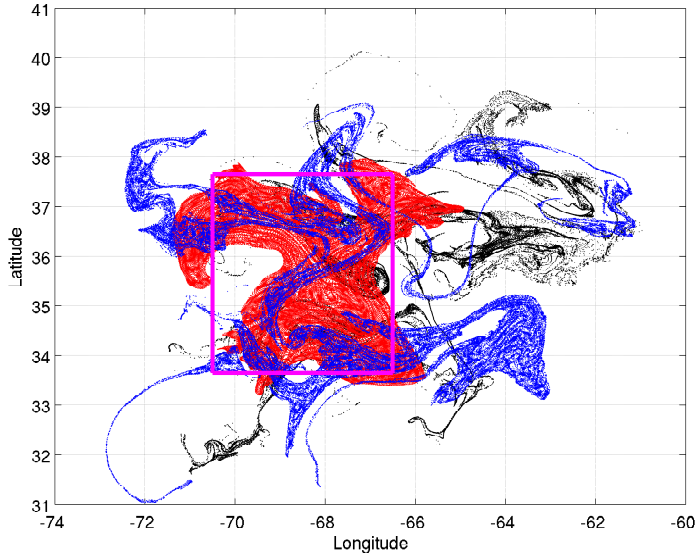}
%\begin{tikzpicture}
%\node[inner sep=0pt] (a) at (0,0){
%\includegraphics[width=0.52\textwidth]{./plots/part_snap.pdf}
%};
%\end{tikzpicture}
\caption{
Position of the $90,000$ drifters after 7 days for the uppermost (black) and the
lowest (red) isopycnals showing the differences in traveled distances as one moves
in the water column.
The positions after 33 days for the lowest isopycnals is shown in blue.
The magenta box indicates the initial positions of the drifter patch.}
\label{fig.traj1}
\end{figure}

\begin{figure}
\centering
\includegraphics{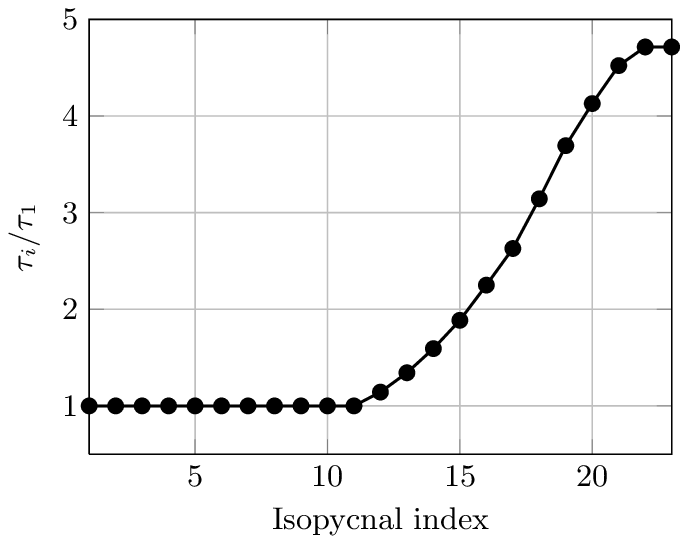}
%\begin{tikzpicture}
%  \begin{axis}[grid=major,width=7.5cm,height=6cm,ymin=0.5,ymax=5,
%   ytick = {1,2,3,4,5},
%   xmin = 1, xmax=23,
%   xlabel={Isopycnal index},
%   ylabel={$\tau_i/\tau_1$},
%   name=main,legend style={at={(0.98,0.02)},anchor=south east}]
%   \addplot[color=black ,mark=*,thick,solid]table[x index=0,y expr=\thisrowno{1}/7]{./data/dtau.dat};
%  \end{axis}
%\end{tikzpicture}
\caption{
Integration time ratio for each isopycnal layer  relative to the near surface.}
\label{fig.traj2}
\end{figure}

\section{Observables: Lagrangian time-averaged basis functions}
\label{sec.observables}

Approaches based on the Koopman operator have been successfully used to analyze dynamical systems that, because of their high-dimensionality,
ill-description and uncertainties, were not well posed for the classical Poincar\'e approach based
on the concept of `dynamic of states' \cite{budisic12}.
The Koopman approach, based on the `dynamics of observables', has demonstrated to be a well-suited alternative
to overcome these difficulties which frequently appear in current engineering and \emph{big data} applications
including fluid dynamics in turbulent flows and Unmanned Aerial Vehicle operations.

In the typical setting, the dynamics on a state space or closed domain $M$ is set by an
iterated map or flow $T$ such that $T: M \rightarrow M$. 
An \emph{observable} on this state space
consists in a function $f: M \rightarrow \mathbb{C}$ where $f$ is an element of
a function space $\mathcal{F}$. 
Conceptually, an \emph{observable} can be described as a probe under the
action of the dynamical system that stores not just the trajectory $p$ but the trace $f(p)$.
The Koopman operator, $U: \mathcal{F} \rightarrow \mathcal{F}$, can then be defined as
\begin{equation}
\left[ U f\right](p)=\left[ f \circ T \right](p)=f(T(p))
\end{equation}
i.e.\ the composition of the observable $f$ and the iterated map $T$.
Typically we only have access to a limited collection of \emph{observables},
$\{ f_1,\dots,f_K\} \subset \mathcal{F}$, that can represent some relevant set for the 
specific problem or a function basis for $\mathcal{F}$.
As such, if $F=(f_1,\dots,f_K)^\textit{T} \in \mathcal{F}$, then the Koopman operator
$U_K: \mathcal{F}_K \rightarrow \mathcal{F}_K$ is defined as:

\begin{equation}
\left[ U_K F\right](p):=\left[
\begin{array}{ccc}
\left[ U f_1\right](p)\\
\vdots\\
\left[ U f_K\right](p)\\
\end{array}
\right].
\end{equation}

In the case of flows with time-dependence, we define a family of Koopman operators
$U_{t_0}^t$ parametrized by the initial time $t_0$ and the final time $T$ in the interval $0 < t \leq T$.
It turns out that time-average of the observables $f(\bx)$ along the trajectories in the finite-time interval can be represented as the average action of this family of operators on the observable
$$
\frac{1}{T-t_0}\int_{t_0}^T U^t_{t_0}f_i(\bx))dt,
$$
and these are in turn described by Dirac measures on trajectories $\nu_\bx^{t_0,T}$.This is very similar to the process in which ergodic measures are associated with infinite time averages in steady flows \cite{Mezic:1994}
(see the Appendix for a more precise description).
 
The ocean model flow fields considered here present two, inter-connected, challenges to 
to directly applying Koopman operator based partitions to the synthetic drifter data. 
First, the domain  of initial conditions, $M$, that we wish to consider 
represents a (small) bounded subset of the  full North Atlantic simulation ($M'$) 
and is  not closed under the flow. 
For finite time trajectories,  we must consider $T: M \rightarrow M'$ where $M \subset M'$, 
i.e.\ $T$ maps trajectories outside the initial domain $M$.
This implies some reconsideration of the shape and spatial extent of the functions typically
used in the necessarily finite basis set.
Secondly, the model velocity field is time-dependent but not periodic and consequently there 
is no single,  well defined, averaging period to separate fields
into mean and fluctuating components. Equivalently, the aperiodic flow-field defines an 
an infinite set of finite-time maps, $T$, parameterized by the finite-time flow interval 
considered.  Since the maps cannot be iterated, both the extent to which trajectories explore the
full domain and  basis-function averages will depend intimately on the averaging time used.

\subsection{Basis functions}

For simplicity of interpretation, we consider a finite set of basis functions given by an 
extension of the Haar wavelet form, $\psi$:
\begin{equation}
\psi_{\mathbf{r},\mathbf{s}}(\mathbf{x})
=
\erf \left( \zeta {(-1)}^{i} \tilde{x}_p \right)
\erf \left( \zeta {(-1)}^{j} \tilde{y}_p \right)
\label{eq.haar}
\end{equation}
where the domain of interest has been rescaled such that 
$M \in [-1/2,1/2]^2$ and
\begin{eqnarray*}
\hat{x}_p=x_p-\frac{s_x r_x}{2}, &\;\;& \hat{y}_p=y_p-\frac{s_y r_y}{2}  \\
i=\left \lfloor \frac{\hat{x}_p}{r_x} \right \rfloor, &\;\;& j=\left \lfloor \frac{\hat{y}_p}{r_y} \right \rfloor \\
\tilde{x}_p=\hat{x}_p-\frac{i}{r_x}-\frac{r_x}{2}, &\;\;& \tilde{y}_p=\hat{y}_p-\frac{j}{r_y}-\frac{r_y}{2} 
\end{eqnarray*}
and $\zeta$ controls the smoothness of the final form of the 2D Haar functional.
The \emph{observable} $f$ is then defined as:

\begin{equation}
f_{\mathbf{r},\mathbf{s}}(\mathbf{x}_p;\tau_i)
=\frac{1}{\tau_i}
\int_{0}^{\tau_i} 
\psi_{\mathbf{r},\mathbf{s}} \left(\mathbf{x_p} \right) dt
\label{eq.f}
\end{equation}
where $\mathbf{x}_p(t)=(x_p(t),y_p(t))$ is the particle position, $\tau_i$ is the integration time period
for isopycnal $i$, $\mathbf{r}$ is the wavelength and $\mathbf{s}$ sets, in each direction, the even and odd 
parity of the function.

To check the sensitivity of the ultimate partition approach to the starting choice of basis set, we will also 
consider harmonic bases,  $\phi$, 
where instead of the wavelength set by $\mathbf{r}$ we use the wavenumber $\mathbf{k}$:

\begin{align}
\phi_{\mathbf{k}} (\mathbf{x}_p)
=
%\mathrm{e}^{2 \pi \mathrm{i} \left( k_x x_p + k_y y_p\right)}\\
\mathrm{e}^{2 \pi \mathrm{i} \left( \mathbf{k} \cdot \mathbf{x}_p\right)}\\
\hat{f}_{\mathbf{k}}(\mathbf{x}_p; \tau_i)
=\frac{1}{\tau_i}
\int_{0}^{\tau_i} \phi_{\mathbf{k}}(\mathbf{x}_p) dt
\label{eq.f2}
\end{align}

The top row in Figure \ref{fig.functions} shows four, Haar-based,  $\psi$ functions defined 
by choosing $(r_x,r_y,s_x,s_y)=(1,2,1,0)$, $(2,1,0,1)$, $(1/2,1/2,0,0)$ and $(1/4,1/4,1,1)$ 
in Eq.~\eqref{eq.haar}.
The smoothing parameter has been set to $\zeta=25$ and the rescaled initial domain $M$ is indicated as a 
black square. The choice $r>1$ allows for basis elements with support extending beyond $M$, potentially
useful in this open flow situation.

\subsection{Dependence on averaging time}

The same basis functions, Lagrangian averaged over times
$\tau_i = 2, 4, 8, 16$ and 32 days, are shown in rows 2-6 respectively of 
Fig.~\ref{fig.functions}. All data corresponds to model layer 12, 
slightly below the surface layers, but
well above the depth of the model Gulf Stream. For the shortest averaging times
$f_{\mathbf{r},\mathbf{s}} \sim \psi_{\mathbf{r},\mathbf{s}}$, and this
tendency is more pronounced for functions with larger spatial extent.
As the averaging time increases, finer scale features emerge in the observables, with each providing Euelrian frame information 
about the temporal evolution of Lagrangian trajectories.  
%with each providing Lagrangian frame information on the position and shape of main cyclonic eddy shown in the Eulerian frame in Fig.~\ref{fig.salvel}c.   
For $\tau_{12}= 8-16$ days, the  large wavelength observables clearly show the existence of a coherent cyclonic eddy 
core detraining particles to the retrograde motion south of the main jet. Similar time averages of the smaller spatial-scale
basis functions, in particular $\psi_{1/4,1/4,1,1}$,  reveal the existence of anti-cyclonic 
flow-features on the jet-side of the main cold-core ring.

The evolution of any observable with averaging time provides information about both the time-scales of the
Lagrangian motion within individual structures and the non-autonomous time-dependence of the structures themselves.
As seen in the bottom row of Fig.~\ref{fig.functions}, long-time averages of smaller spatial scale functions
converge to uniform, zero-mean fields while averages of spatially extended modes develop very fine-scale spatial 
filamentation. In non-autonomous, finite-time flows, the choice of averaging time for the observables is 
essentially equivalent to the choice of integration time in  Finite-Time or Direct Lyapunov Exponent, or Mesohyperbolicity 
based LCS methods (see, for example, the results shown in Figs.~\ref{fig.mh} and  \ref{fig.dle}).

The choice of optimal averaging time is subjective and depends on the time dependence of the flow and the ultimate
degree of spatial feature resolution sought. In the following we seek to combine information from a fixed set of 
observables to partition the full 3D+1 flow field. Fixing the smallest spatial scale basis functions to be 
$\mathbf{r}=(1/4,1/4)$ (nominally probing features with spatial scales 1/4 of the horizontal domain of interest), then the
results in Fig.~\ref{fig.functions} indicate that the choice $\tau_{12} = 8$ days is both long enough 
to allow the emergence of structure information in this observable and short enough to avoid homogenization. 
This choice of integration time
also reduces the degree of filamentation in the larger wavelength observables. The temporal evolution of the eddy
shown in Fig.~\ref{fig.salvel} indicates that $\tau_{12} = 8$ days is an intermediate time scale for the flow
field. Particles sample several rotation periods of the eddy and travel 
significant distances in the jet during this time period where the structures themselves have evolved more slowly.

\begin{figure}
\centering
\includegraphics{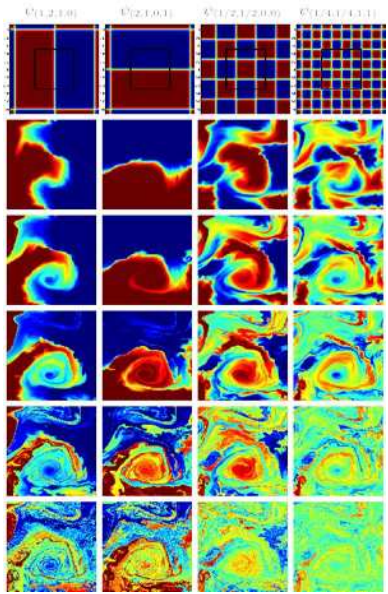}
\caption{
Top panel, left to right: Haar functions for $(r_x,r_y,s_x,s_y)=
 (1,2,1,0)$,
$(2,1,0,1)$,
$(1/2,1/2,0,0)$
and
$(1/4,1/4,1,1)$.
Panels in rows 2-6, top to bottom: Lagrangian averages at
$\tau=2, 4, 8, 16, 32$ for isopycnal $i=12$ and day $t_0=80$.
}
\label{fig.functions}
\end{figure}

\subsection{Depth Dependence}

With the temporal rescaling shown in Fig.~\ref{fig.traj2},  a single horizontal basis function can be used to
explore the evolution of the coherent features throughout the water column. With $\tau_{12}$ set to
8 days, the results for the averages of $\psi_{2,1,0,1}$ at four different isopycnal layers 
are shown in Fig.~\ref{fig.depth}. Layer number increases from (a) to (d).

The observable indicates the  persistence of the cyclonic feature with increasing depth. While there is a
clearly defined core near the surface, this feature is increasingly filamented in the lower layers. The nearly
elliptical surface expression evolves into a arrowhead shape at layer 18 indicating
entraining lobes of cyclonic trajectories along the southern edge of the ring (red) and detraining
trajectories (yellow and green) from the core. As expected, extreme vorticity 
gradients associated with the Gulf Stream in the upper layers  produce a strong transport barrier. With
increasing depth, this barrier is weakened as revealed by the presence of cross-jet transport pathways in the
time averaged function.  These pathways appear to be connected to the appearance, at intermediate layers, of
smaller-scale anti-cyclonic features on the jet-side of the cyclone. Overall, the single observable indicates
evolution from a relatively isolated cyclonic feature at the surface to a complex multi-pole structure at depth.

\begin{figure}
\centering
\includegraphics{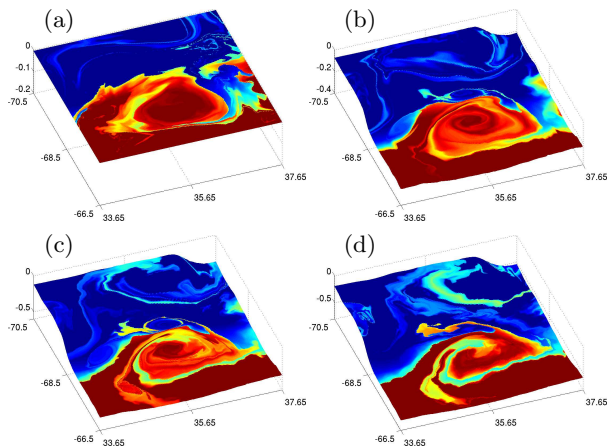}
%\begin{tikzpicture}
%\node[inner sep=0pt] (a) at (0,0){
%\includegraphics[width=0.22\textwidth]{./plots/lev05.png}
%}node at (-1.3,1.25) {(a)};
%%\end{tikzpicture}
%%\begin{tikzpicture}
%\node[inner sep=0pt] (b) at (4,0){
%\includegraphics[width=0.22\textwidth]{./plots/lev14.png}
%}node at (2.7,1.25) {(b)};
%%\end{tikzpicture}
%%\\
%%\begin{tikzpicture}
%\node[inner sep=0pt] (c) at (0,-3){
%\includegraphics[width=0.22\textwidth]{./plots/lev16.png}
%}node at (-1.3,-1.75) {(c)};
%%\end{tikzpicture}
%%\begin{tikzpicture}
%\node[inner sep=0pt] (d) at (4,-3){
%\includegraphics[width=0.22\textwidth]{./plots/lev18.png}
%}node at (2.7,-1.75) {(d)};
%\end{tikzpicture}
\caption{3D view of $\psi^k_{2,1,0,1}$ on day 80 for four different isopycnal layers
$k=$5, 14, 16 and 18.}
\label{fig.depth}
\end{figure}

\section{Partition}
\label{sec.partition}

As shown above, the  Lagrangian statistics of individual basis functions provide detailed information on the location,
spatial extent and advective connections between coherent structures in the flow field.  The goal here is to
combine the complimentary information provided by a collection of functions to partition the initial condition
domain $M$ into discrete sets within which trajectories possess `most similar' statistics over the set of observables. 

The finite-time/open-flow nature of the oceanic transport problem poses questions for application of standard
ergodic partition approach. For a dynamical system defined by a (Poincar\'e) map on a closed domain, 
ergodic theory allows one to compute the convergence properties of time averages of observable functions 
(see discussion in \cite{levnajic10}). In the ocean flow, however, there is no
well defined, long-time average and, even if there were, one is typically not interested in transport
on such long times.
In the appendix we develop the mathematical formalism of
Finite-Time Ergodic (FiTEr) partitions.
A FiTEr partition associated with a finite set of functions $\{f_i\}, i=1,...,n$ is the partition of $\cD$ into
joint level sets of $\tilde f_i^{t_0,T}(\bx), i=1,...n$.
We approximate these joint level sets using clustering algorithms from machine learning.
\bigskip

Using the Haar-like functions, $\psi_{\mathbf{r},\mathbf{s}}$, as basis-set elements $\{f_i\}, i=1,...,n$ and
accounting for the lack
of orthogonality on $M$ when $r>1$, we will consider sets of functions such as those given in Table \ref{tab.f}. 
This 24 dimensional basis set is obtained by selecting $\mathbf{r} = (1/2,1,10)^2$ and eliminating the
\textit{mean} mode, $r_x=r_y=10$ and any redundancies produced by symmetry. Such collections are the Haar
equivalents of standard discrete Fourier bases and similar sets can be produced by varying the choice of $\mathbf{r}$.

\begin{table}
\centering
\begin{tabular}{cc|c|c|c|c|c|}
\cline{3-7}
 & & $r_x=10$ & \multicolumn{2}{c|}{$r_x=1$} & \multicolumn{2}{c|}{$r_x=1/2$}\\ \cline{3-7}
 & & $s_x=0$ & $s_x=0$ & $s_x=1$ & $s_x=0$ & $s_x=1$\\ \hline
\multicolumn{1}{|c|}{$r_y=10$} & $s_y=0$ & - & 1 & 2 & 3 & 4\\ \hline
%$r_y=10$ & \multicolumn{1}{|c|}{$s_y=0$} & - & 1 & 2 & 3 & 4\\ \hline
\multicolumn{1}{|c}{
\multirow{2}{*}{$r_y=1$}}  & \multicolumn{1}{|c|}{$s_y=0$} & \textcolor{blue}{5}  &
                        6  & 7  & 8  & \textcolor{red}{9} \\ \cline{3-7}
\multicolumn{1}{|c}{}      & \multicolumn{1}{|c|}{$s_y=1$} & \textcolor{blue}{10} &
                        11 & 12 & 13 & \textcolor{red}{14} \\ \hline
\multicolumn{1}{|c}{
\multirow{2}{*}{$r_y=1/2$}}& \multicolumn{1}{|c|}{$s_y=0$} & \textcolor{blue}{15} &
                        16 & 17 & 18 & \textcolor{red}{19} \\ \cline{3-7}
\multicolumn{1}{|c}{}      & \multicolumn{1}{|c|}{$s_y=1$} & 20 & 21 & 22 & 23 & 24 \\ \hline
\end{tabular}
\caption{Haar function set. 3D projections on  the functions in blue and red are shown in Fig.~\ref{fig.cluster3D}.}
\label{tab.f}
\end{table}

Given a set of $N$ functions, each of the 90,000 drifters advected in each of the
$L=23$ isopycnals, i.e.\ $n=90,000 \times 23 = 207\cdot10^4$ trajectories, produces a single point in the
$\mathbb{R}^N$ space of observables. 
Using the set of 24 functions given in Tab.~\ref{tab.f}, 
samples of the cloud of $n$ data points, indicated by small colored points, are shown in Fig.\ref{fig.cluster3D} in two,
3D projections of $\mathbb{R}^{24}$. The specific function triplets defining the projections,
$(r_x,r_y,s_x,s_y)=
(10,1,0,0)$,
$(10,1,0,1)$,
$(10,1/2,0,0)$ and 
$(r_x,r_y,s_x,s_y)=
(1/2,1,1,0)$,
$(1/2,1,1,1)$,
$(1/2,1/2,1,0)$ are marked in blue and red in Tab.~\ref{tab.f}. 

The distribution of data in clearly organized and non-uniformly distributed in $\mathbb{R}^{24}$.
Fig.~\ref{fig.cluster3D}a, the projection on quasi-meridional modes ($r_x = 10$), shows alignment
of the data cloud on distinct planes. Values of the odd, $f_{10,1,0,1}$, mode concentrate almost 
entirely  on $\pm1$ except where the even mode $f_{10,1,0,0} = 1$. The projection on observables with smaller 
spatial scale (Fig.~\ref{fig.cluster3D}b), while more complex, continues to show organization of the data on
lower dimensional subspaces of $\mathbb{R}^{24}$.

\begin{figure}
\centering
\includegraphics{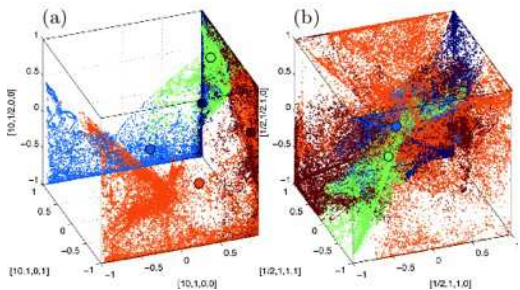}
%\begin{tikzpicture}
%\node[inner sep=0pt] (a) at (0,0){
%\includegraphics[width=0.23\textwidth]{./plots/f3d01.png}
%}node at (-1.3,2.1) {(a)};
%%\end{tikzpicture}
%%\begin{tikzpicture}
%\node[inner sep=0pt] (b) at (4,0){
%\includegraphics[width=0.23\textwidth]{./plots/f3d02.png}
%}node at (2.7,2.1) {(b)};
%\end{tikzpicture}
\caption{
3D projection of the 24D partition for functions
$(10,1,0,0)$,
$(10,1,0,1)$ and
$(10,1/2,0,0)$ (left) and
$(1/2,1,1,0)$,
$(1/2,1,1,1)$ and
$(1/2,1/2,1,0)$ (right).
The large spheres indicate the location of each cluster centroid.}
\label{fig.cluster3D}
\end{figure}

\subsection{\lowercase{$k$}-\lowercase{means} clustering}

We seek to define a partition of $M$ based on the distribution of the $n$ data points in the 
$N$ dimensional space of observables. In other words, we seek to group the $n$ points into
$m$ clusters, such that the statistics defined by the $N$ observables in each cluster are closer
to each other than to those in other clusters. A standard data-mining cluster algorithm is provided
by the $k$-means method. 
The algorithm aims to partition the data points into $m$ clusters so that the within-cluster
sum of the squared Euclidean distance between each point and the cluster centroid is minimized.
Specifically, we use a standard $k$-means implementation provided by the Statistics and Machine Learning
MatLab toolbox.
The algorithm is initialized using a subsample of the data, randomly selected according to a previously fixed
seed for reproducibility. 
To ensure that the results are independent of the initialization, the procedure is repeated several times with
different initial estimatess for the centroids.

Results of the $k$-means algorithm for $m=5$ clusters are shown in Fig.~\ref{fig.cluster3D} where the individual data
points have been color-coded based on cluster index and the projections of the $m$ cluster centroid locations are
marked by large, solid spheres. The corresponding partition of $M$ in physical space for isopycnal
layer 15 is given in Fig.~\ref{fig.clusters}a.

\begin{figure}
\centering
\includegraphics{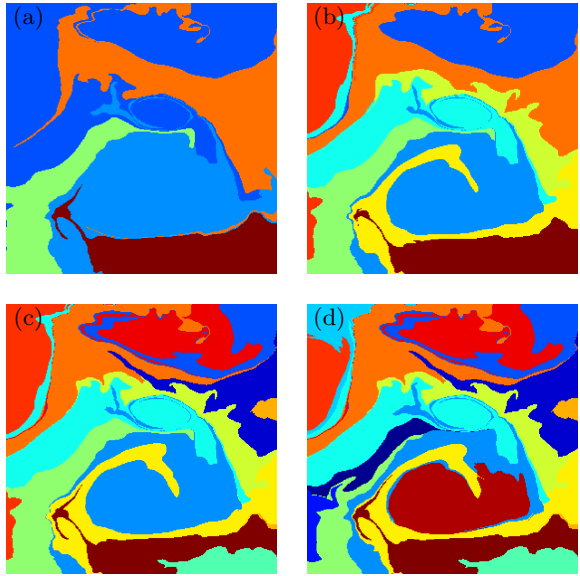}
%%\begin{tabular}{ccc}
%\begin{tikzpicture}
%\node[inner sep=0pt] (a) at (0,0){
%\includegraphics[width=0.20\textwidth]{./plots/clust_05_lev_15.png}
%}node at (-1.5,1.6) {(a)};
%%\end{tikzpicture}
%%&\hspace{0.1in}&
%%\begin{tikzpicture}
%\node[inner sep=0pt] (b) at (4,0){
%\includegraphics[width=0.20\textwidth]{./plots/clust_09_lev_15.png}
%}node at (2.5,1.6) {(b)};
%%\end{tikzpicture}
%%\\
%%\begin{tikzpicture}
%\node[inner sep=0pt] (c) at (0,-4){
%\includegraphics[width=0.20\textwidth]{./plots/clust_13_lev_15.png}
%}node at (-1.5,-2.4) {(c)};
%%\end{tikzpicture}
%%&\hspace{0.1in}&
%%\begin{tikzpicture}
%\node[inner sep=0pt] (d) at (4,-4){
%\includegraphics[width=0.20\textwidth]{./plots/clust_17_lev_15.png}
%}node at (2.5,-2.4) {(d)};
%\end{tikzpicture}
%%\end{tabular}
\caption{
Top to bottom ,left to right: ergodic partition using different number of clusters $m=5,9,13,17$.
Layer is $k=15$, day is 80 and the dimension of the function space is $n=24$.}
\label{fig.clusters}
\end{figure}

\subsection{Convergence properties}

The ultimate $k$-means based partition of $M$ depends on the number of clusters sought, $m$, and both the
dimensionality, $N$, and the specific spatial form of the input basis set of observables. 

Fig.~\ref{fig.clusters} shows the dependence, for fixed isopycnal layer 15, of the partition on the number of 
clusters used. Each frame corresponds to partitions of the $N=24$ Haar basis set with $m = 5,9,13,17$.
The main cyclonic eddy, previously identified in individual observables, is clearly delineated in each partition.
Each partition also clearly identifies coherent anti-cyclonic sets on the jet side of the cyclone. 
Increasing the number of clusters simply refines the partition by dividing existing features into distinct 
subsets.
For $m=5$, the main cyclonic eddy (light blue) has both a thin entraining lobe surrounding the anti-cyclone and
a larger detraining region to the southwest. For $m=9$ clusters, the detraining region
forms a new (yellow) cluster and the definition of the cold-core 'eddy' is refined. Similar refinement of
the detrainment region and the eddy core takes place when $m$ increases from 9 to 13.
In this sense, for fixed $N$, the resulting partition is stable with respect to the number of clusters;
$m$ effectively sets the spatial scale of the resulting coherent sets. 
Again, for simplicity of interpretation of the time dependent, three-dimensional (3D+1) results, we take $m=9$
in what follows. 

\bigskip

\begin{figure}
\centering
\includegraphics{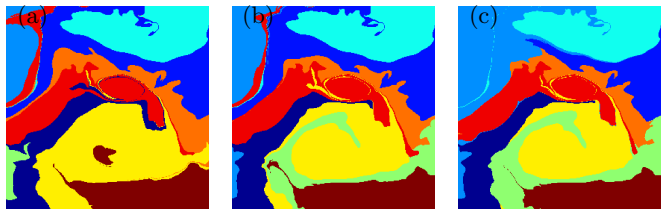}
%\begin{tikzpicture}
%\node[inner sep=0pt] (a) at (0,0){
%\includegraphics[width=0.15\textwidth]{./plots/junk1.png}
%}node at (-1.0,1.2) {(a)};
%%\end{tikzpicture}
%%\begin{tikzpicture}
%\node[inner sep=0pt] (b) at (3,0){
%\includegraphics[width=0.15\textwidth]{./plots/junk2.png}
%}node at (2.0,1.2) {(b)};
%%\end{tikzpicture}
%%\begin{tikzpicture}
%\node[inner sep=0pt] (c) at (6,0){
%\includegraphics[width=0.15\textwidth]{./plots/junk3.png}
%}node at (5.0,1.2) {(c)};
%\end{tikzpicture}
\caption{Ergodic partition at isopycnal $i=15$ using different function space dimensions $n=8,24,48$.}
\label{fig.funcsize}
\end{figure}

The function set presented in Tab. \ref{tab.f} can be modified by altering the choice of $\mathbf{r}$ used
to define the basis.
Fig. \ref{fig.funcsize} shows a comparison at isopycnal $i=15$ of the $m=9$ partition for 8, 24 and 48 
dimensional function sets generated by choosing $\mathbf{r} = (1,10), (1/2,1,10)$ and $(1/4,1/2,1,10)$. As such the 
8 dimensional set is a proper subset of the 24 dimensional set which in turn is a proper subset of the 48 
dimensional set. Convergence of the large scale structures with the number of Haar basis elements is apparent.
The results shown in Fig. \ref{fig.functions} indicate that even moderate spatial scale observables are homogenized by 
representative temporal averages. As such, including smaller scale
features in the functional basis has a nominal effect on the main coherent features identified by the 
global $k$-means partition. 

%\begin{table}
%\centering
%\begin{tabular}{c|cc|cc|cc|}
%\cline{2-7}
% & \multicolumn{2}{c|}{$k_x=0$} & \multicolumn{2}{c|}{$k_x=1/2$} & \multicolumn{2}{c|}{$k_x=1$} \\ \cline{2-7} \hline
%\multicolumn{1}{|c|}{$k_y=0$} &
%\multicolumn{2}{c|}{-} & \multicolumn{1}{c|}{1,$\re$} & 2,$\im$ & \multicolumn{1}{c|}{3,$\re$} & 4,$\im$ \\ \hline
%\multicolumn{1}{|c|}{$k_y=1/2$}
%& \multicolumn{1}{c|}{5,$\re$} & 6,$\im$ & \multicolumn{1}{c|}{7,$\re$} & 8,$\im$ & \multicolumn{1}{c|}{9,$\re$} & 10,$\im$ \\ \hline
%\multicolumn{1}{|c|}{$k_y=1$}
%& \multicolumn{1}{c|}{11,$\re$} & 12,$\im$ & \multicolumn{1}{c|}{13,$\re$} & 14,$\im$ & \multicolumn{1}{c|}{15,$\re$} & 16,$\im$ \\ \hline
%\end{tabular}
%\caption{Harmonic function set.}
%\label{tab.f2}
%\end{table}

As a further test of the sensitivity of the overall partition to the details of the space of observables, we also 
consider a standard, harmonic plane wave basis (see Eq. \ref{eq.f2}).
A comparison of the $m=9$ $k$-means partition for 24 
dimensional Haar and 16 dimensional harmonic set is shown in Fig.~\ref{fig.funcsize2}.
The harmonic set corresponds to $\mathbf{k}=(0,1/2,1)^2$ where the $k_x=k_y=0$ has been discarded. Note that each pair
of wavenumber values ($k_x,k_y$) has associated two solutions for the real and imaginary parts of the harmonic function.

Again, the 
overall partition, especially the shape and location of the coherent cold-core eddy, is remarkably insensitive to the 
functional form of the basis.
The major difference between the harmonic functions and the Haar functions is the density distribution of the range of 
the functions. By construction, the range of the Haar functions is concentrated at $\pm 1$. As seen in the figure, the broader 
distribution of the range of the harmonic functions produces, for fixed averaging time, a more detailed, finer spatial 
scale, picture of the flow field.

\begin{figure}
\centering
\includegraphics{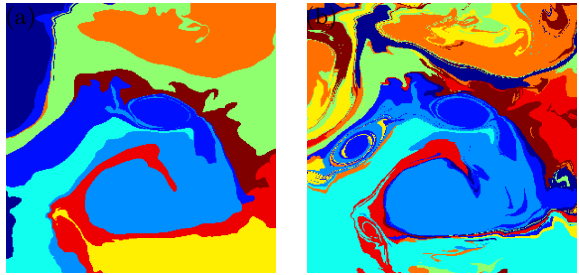}
%\begin{tikzpicture}
%\node[inner sep=0pt] (a) at (0,0){
%\includegraphics[width=0.2\textwidth]{./plots/haar.png}
%}node at (-1.6,1.6) {(a)};
%%\end{tikzpicture}
%%\begin{tikzpicture}
%\node[inner sep=0pt] (b) at (4,0){
%\includegraphics[width=0.2\textwidth]{./plots/harm.png}
%}node at (2.4,1.6) {(b)};
%\end{tikzpicture}
\caption{Comparison of the 24 dimensional Haar basis (left) and 16 dimensional harmonic basis (right).}
\label{fig.funcsize2}
\end{figure}

\subsection{Depth dependence: Comparison to LCS metrics}

For a given initialization time, the partitions result from $k$-means clustering of the full Lagrangian data set
across all 23 isopycnal levels. The depth dependence of the coherent sets produced by the $m=9$, $N=24$ partition 
is shown at three representative isopycnal levels in the first columns of  Figs.~\ref{fig.mh} and \ref{fig.dle}.
Here the model trajectories were launched on day 74. 

The $k$-means partition provides considerable refinement of the depth dependence of the main cyclonic 
structure shown previously, for a single observable initialized on day 80, in Fig.~\ref{fig.depth}. 
In the near-surface, the cold-core ring at day 74 (light blue)  is entraining fluid along a thin filament on the
southern side of the jet and the transport geometry is similar to that of 'eddy pinch-off' processes previously
observed in simplified, analytic and low-resolution models \cite{DutPal94,poje99} of eddy-jet interactions.  
At mid-depth, the entrainment from the jet is impeded by the appearance of a binary pair of smaller scale, 
anti-cylonic eddies (medium blue) in the region between the jet and the cyclone. The main cyclonic feature is
now filamented by detraining lobes of fluid (light blue, yellow). At depth, the main core continues to shrink in size,
developing into a arrowhead or mushroom shape usually associated with vortex dipole configurations.

\bigskip

Also shown in Figs.~\ref{fig.mh} and \ref{fig.dle} are direct comparisons between the partitions derived from $k$-means 
clustering of Lagrangian averaged observables and more traditional Lagrangian Coherent Structure identification metrics,
namely the forward time \textit{Mesohyperbolicity field} \cite{Mezic2010} and the field of forward-time 
\textit{Direct Lyapunov Exponents} \cite{Haller2002}. 
Both metrics have been widely used to analyze Lagrangian oceanographic transport \cite{Waugh2008,Beron2008,Mezic2010}.

Following \cite{Haller2002}, the maximal finite-time Direct Lyapunov Exponent is given by 
\begin{equation}
 DLE(\mathbf{x}_0,t_0,\tau) =  \frac{1}{2\tau} \log{\sigma_t(\mathbf{x}_0}) 
\end{equation}
where $\sigma_\tau(\mathbf{x}_0)$ is the maximal eigenvalue of the Cauchy-Green strain tensor, $C^{\tau}(\mathbf{x}_0)$, 
the symmetrized gradient of the flow with respect to initial conditions.

Alternatively, LCS features can be extracted by constructing the mesohyperbolicity field, $K(\mathbf{x}_0)$,
over initial conditions defined by 
\begin{equation}
K(\mathbf{x}_0,t_0,\tau)=\det{\nabla \mathbf{u}^*(\mathbf{x}_0,t_0,\tau)},
\end{equation}
where $\mathbf{u}^*(\mathbf{x}_0,t_0,\tau)$ is the average Lagrangian velocity on $\left[t_0,t_0+\tau \right]$ starting from
$\mathbf{x}_0$.
Note that, when $t \rightarrow t_0$, the mesohyperbolicity reduces to the well-known, Eulerian Okubo-Weiss criterion.

Like the Lagrangian averaged observables, both the DLE and mesohyperbolicity metrics depend explicitly on  trajectory
integration time. This dependence is shown across columns for the 3 ispocynal layers in 
Figs.~\ref{fig.mh} and \ref{fig.dle}. 

The distinct differences between the cluster-based partition, arrived at by globally searching for coherent sets across all 
trajectories, and the LCS approaches based entirely on local trajectory information are apparent in the comparisons.
The highly time-dependent velocity fields, generated by a model resolving the largest \textit{sub}mesoscale flow
dynamics, produces complex spatial $DLE$ and $K$ fields for any reasonable integration period.
While the main coherent structures, such as the cold-core eddy and associated smaller-scale anti-cyclones are revealed by
the maps of both $DLE$ and $K$, the boundaries of these structures and the advective pathways connecting them are obscured 
by the rich complexity and filamentation of the LCS fields. With a broad range of energetic spatial scales, advective transport 
both within and between the time-dependent mesoscale structures is modulated by the presence of smaller features with shorter
lifetimes resulting in extreme tangling of the finite-time hyperbolic sets organizing the advection. 

\begin{figure*}[ht]
\centering
\includegraphics{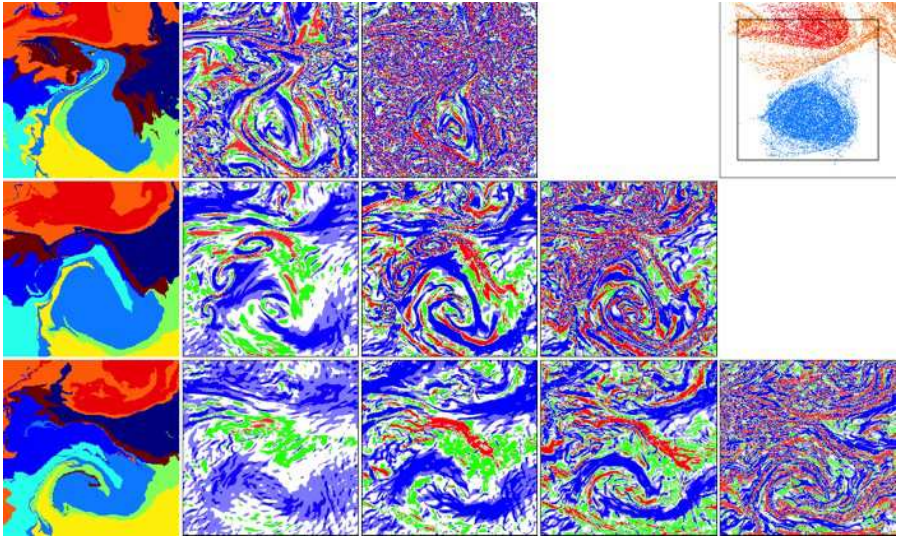}
%%\begin{tabular}{ccccc}
%\begin{tikzpicture}
%\node[inner sep=0pt] (plot1) at (0,0)
%{\includegraphics[width=0.195\textwidth]{./plots/Time_005_lev_01.png}};
%\node[inner sep=0pt] (plot1) at (3.56,0)
%{\includegraphics[width=0.195\textwidth]{./plots/mesoh/trim0000.png}};
%\node[inner sep=0pt] (plot1) at (7.12,0)
%{\includegraphics[width=0.195\textwidth]{./plots/mesoh/trim0001.png}};
%\node[inner sep=0pt] (plot1) at (14.24,0)
%{\includegraphics[width=0.195\textwidth]{./plots/newigor1_t.png}};
%%
%\node[inner sep=0pt] (plot1) at    (0,-3.56)
%{\includegraphics[width=0.195\textwidth]{./plots/Time_005_lev_15.png}};
%\node[inner sep=0pt] (plot1) at (3.56,-3.56)
%{\includegraphics[width=0.195\textwidth]{./plots/mesoh/trim0002.png}};
%\node[inner sep=0pt] (plot1) at (7.12,-3.56)
%{\includegraphics[width=0.195\textwidth]{./plots/mesoh/trim0003.png}};
%\node[inner sep=0pt] (plot1) at (10.68,-3.56)
%{\includegraphics[width=0.195\textwidth]{./plots/mesoh/trim0004.png}};
%%
%\node[inner sep=0pt] (plot1) at    (0,-7.12)
%{\includegraphics[width=0.195\textwidth]{./plots/Time_005_lev_20.png}};
%\node[inner sep=0pt] (plot1) at (3.56,-7.12)
%{\includegraphics[width=0.195\textwidth]{./plots/mesoh/trim0005.png}};
%\node[inner sep=0pt] (plot1) at (7.12,-7.12)
%{\includegraphics[width=0.195\textwidth]{./plots/mesoh/trim0006.png}};
%\node[inner sep=0pt] (plot1) at (10.68,-7.12)
%{\includegraphics[width=0.195\textwidth]{./plots/mesoh/trim0007.png}};
%\node[inner sep=0pt] (plot1) at (14.24,-7.12)
%{\includegraphics[width=0.195\textwidth]{./plots/mesoh/trim0008.png}};
%%\end{tabular}
%\end{tikzpicture}
\caption{First column: Day 74 partition derived from clustering for layers 1 (top), 15 (middle) and 20 (bottom). 
Averaging times used are $\tau$ = 7, 12 and 24 days respectively. Second column: Mesohyperbolicity for the same three
layers using averaging time $\tau$ = 3 days. Third column: Same for $\tau$ = 7 days. Fourth column: $\tau$ = 12 days. Fifth column:
$\tau$ = 24 days. 
Final figure in first row shows 7 day trajectories for the red, orange and blue clusters colored 
accordingly.}
\label{fig.mh}
\end{figure*}

\begin{figure*}[ht]
\centering
\includegraphics{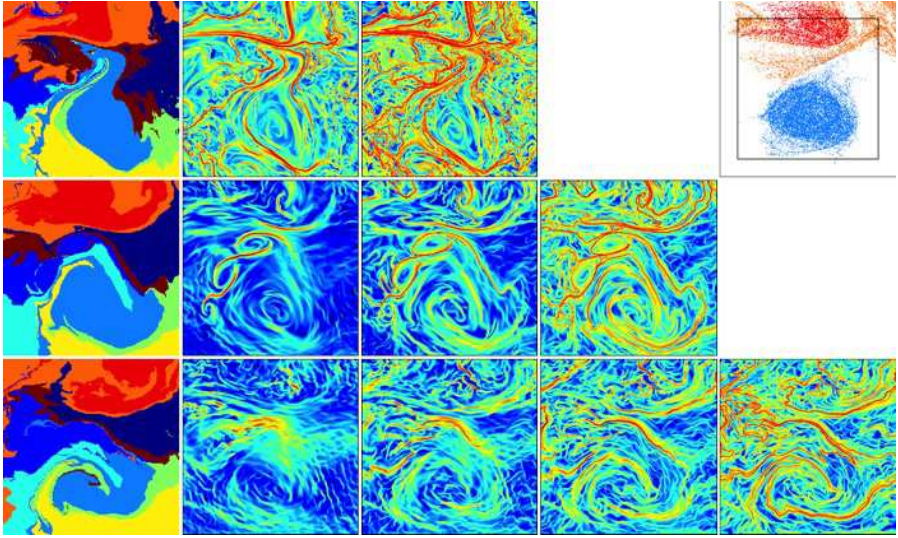}
\caption{Same Fig.~\ref{fig.mh} but for  DLE fields instead of Mesohyperbolicity.}
\label{fig.dle}
\end{figure*}

To illustrate the  connection between the finite-time partition at the surface and particle behavior, trajectories associated with 
3 clusters are shown in the far right top panel in Fig. \ref{fig.mh} and \ref{fig.dle}. 
The clusters chosen correspond to the main cyclone (light blue partition) and two  groups of initial positions in the jet (red and orange). 
Particles were released on day 74 and advected over the averaging period (7 days).
The results clearly show that all initial conditions associated with the light-blue cluster, even those in the entraining lobe, exhibit
strong cyclonic circulation during the averaging period.
In the jet region, the clustering algorithm distinguishes initial conditions with strongly anti-cyclonic swirl (red) from those that advect with
the jet without being entrained into an eddy (orange).

\section{3D+1 reconstruction: Temporal evolution of partition}
\label{sec.reconstruction}

As seen above, even when working in an isopycnal setting, piecing together the depth dependence of 
coherent features from discrete vertical sets of $DLE$ or $K$ fields is difficult. In contrast,
the $k$-means partitions, as defined here, directly incorporate information from all layers and
produce fully three-dimensional partitions at any initialization time.  To investigate the temporal
evolution of the features, one need only to reproduce the analysis for a series of launch dates.

Fig.~\ref{fig.3dp1} shows the time and depth dependence of the $m=9$ and $N=24$ partition at
four different isopycnal layers (rows) for five different initialization times (columns).

The cyclonic eddy (light blue) partition deforms and rotates under the action of the jet
(red/orange for $t\le74$ days and orange/green for $t>74$).
The results for the top isopycnal reveal that for $t>80$ days the upper lobe of the cold core weakens
indicating a suppression of the entrainment that is complete at day 86. 
This event is accompanied by the emergence of anticyclonic structures (dark blue and cyan partitions)
in the upper and lower sides of the cyclonic eddy that exhibit respective entraining lobes as observed at day 90.
While the spatial extension of the anticyclonic eddy in the north side grows moderately in the vertical
at days 70-80, later launch times partitions reveal much stronger reinforcements.
For instance, at day 86 the entraining lobe of the dark blue partition strengthens as depth increases and at day 90
the structure resembles that of a dipole with comparable sizes of the partitions for the cyclone and the anticyclone.
On the other hand, the second anticyclone structure (cyan) clearly observed at the surface on day 90
vanishes rapidly in the vertical. 

In general, the results for the vertical evolution the cyclone indicates a reduction in its size in the vertical.
On days 70-80 and at the deepest isopycnal level, the partion of the cold-core eddy exhibits
a characteristic arrowhead shape generated by the relatively small anticyclonic features.
The results suggest that later on these anticyclonic structures grow bottom-up and eventually become
large enough to generate a dipole structure clearly observed at day 90.

The cold-core ring that spanned across the whole water column for 10 days, \emph{decays} thereafter indicating the
cyclone vertical shrinking.
This process is quantified by estimating this structure volume and the vertical location of its centroid
as shown in Fig.~\ref{fig.cent_vol}.
The results show that as time evolves the cold cyclone size decreases by 50\%
going from 2,000 to 1,000 cubic kilometers over the span of 20 days.
As expected, this volume reduction is accompanied by a reduction in the vertical
centroid location.

\begin{figure}
\centering
\includegraphics{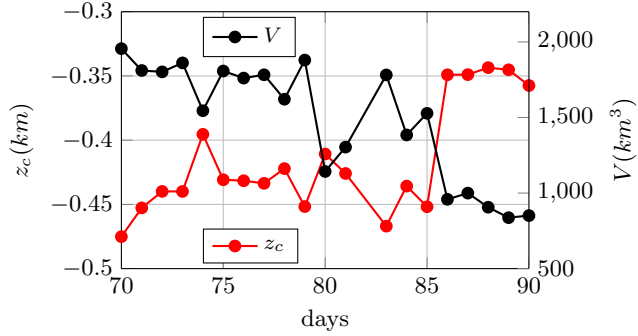}
%\begin{tikzpicture}
%  \begin{axis}[grid=major,width=7cm,height=5cm,xlabel={days},ylabel={$z_c (km)$},axis y line*=left,
%    xmin=70,xmax=90,legend style={at={(0.2,0.02)},anchor=south west},ymin=-0.5,ymax=-0.3]
%    \addplot [color=red,mark=*,thick,solid]table[x expr=\thisrowno{0}+69,y index=3]{./data/z_centroid.txt};
%    \addlegendentry{$z_c$};
%  \end{axis}
%  \begin{axis}[width=7cm,height=5cm,hide x axis, axis y line*=right,
%       ylabel={$V (km^3)$},xmin=70,xmax=90,legend style={at={(0.2,0.98)},anchor=north west},ymin=500,ymax=2200]
%    \addplot [color=black,mark=*,thick,solid]table[x expr=\thisrowno{0}+69,y index=1]{./data/volume.txt};
%    \addlegendentry{$V$};
%  \end{axis}
%%\end{tikzpicture}
\caption{Temporal evolution of the cold-core centroid vertical location (red) and volume (black).}
\label{fig.cent_vol}
\end{figure}

\begin{figure*}[ht]
\centering
\includegraphics{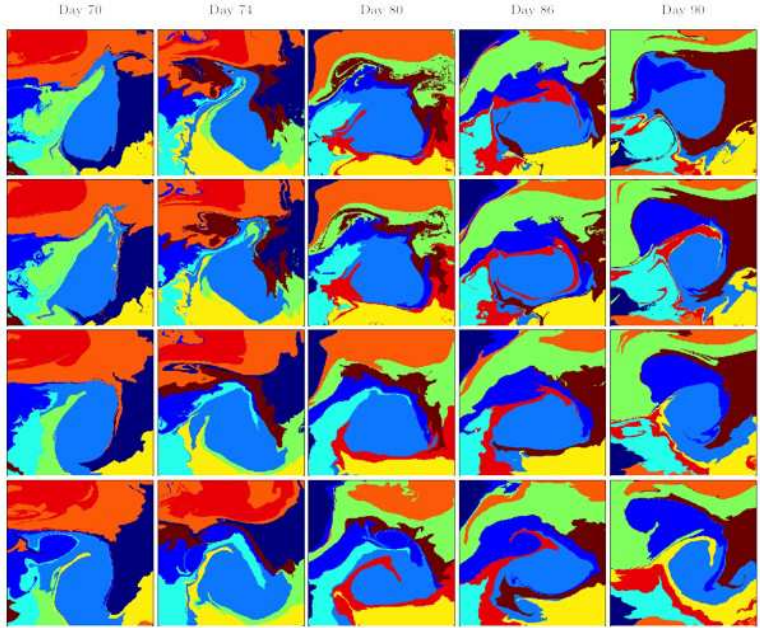}
\caption{Temporal and depth dependence of the $k$-means partition for $m=9$ and $N=24$.
Columns: days 70, 74, 80, 86 and 90. Rows: isopycnals 1, 8 , 12 and 15.}
\label{fig.3dp1}
\end{figure*}

A full picture of the 3D+1 flow evolution in shown in Fig.~\ref{fig.new2}
that contains three-dimensional snapshots of two selected structures at days 73, 78, 84 and 90.
Both the vertical shrinking of the cyclonic structure (light blue) and the bottom-up
surface emergence of the anticyclone
(dark blue) are clearly depicted.
These images reveal the extent of both structures and the interaction between them as time evolves.
Specifically, it is possible to observe in detail how the anticyclonic core, initially at the deep northwest side
of the cold ring, travels in the clockwise direction while resurfaces and impedes the entrainment into the cyclone
which then vertically shrinks as shape changes under the action of the jet stream.

\begin{figure*}[ht]
\centering
\includegraphics{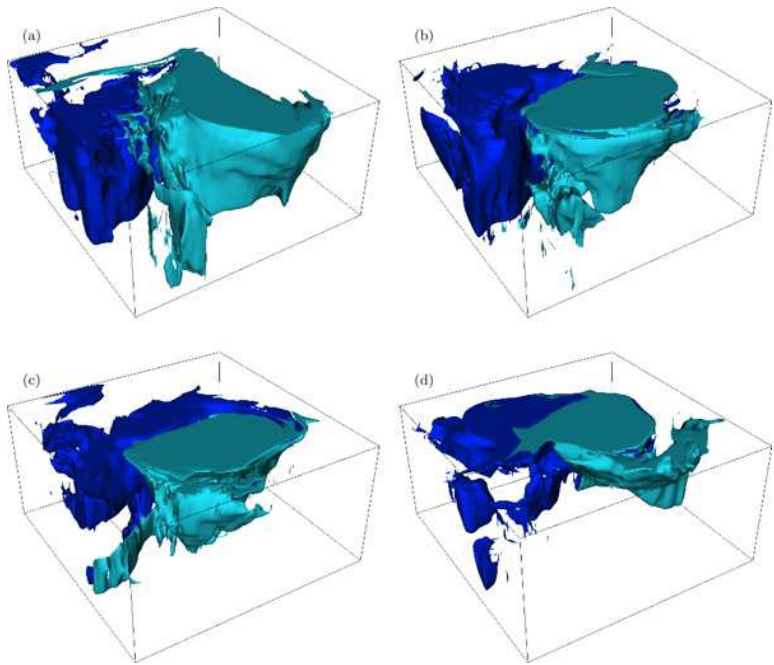}
%%\begin{tabular}{cc}
%%
%\begin{tikzpicture}
%\node[inner sep=0pt] (a) at (0,0){
%\includegraphics[width=0.45\textwidth]{./plots/trim0003.png}
%}node at (-3.5,2.75) {(a)};
%%\end{tikzpicture}
%%&
%%\begin{tikzpicture}
%\node[inner sep=0pt] (b) at (8.5,0){
%\includegraphics[width=0.45\textwidth]{./plots/trim0008.png}
%}node at (5,2.75) {(b)};
%%\end{tikzpicture}
%%\\
%%\begin{tikzpicture}
%\node[inner sep=0pt] (c) at (0,-7.5){
%\includegraphics[width=0.45\textwidth]{./plots/trim0014.png}
%}node at (-3.5,-4.75) {(c)};
%%\end{tikzpicture}
%%&
%%\begin{tikzpicture}
%\node[inner sep=0pt] (d) at (8.5,-7.5){
%\includegraphics[width=0.45\textwidth]{./plots/trim0020.png}
%}node at (5,-4.75) {(d)};
%\end{tikzpicture}
%%\\
%%\end{tabular}
\caption{Temporal evolution of the three-dimensional partitions for the cold-core eddy (light blue) and
the anticyclonic feature (dark blue) at days 73, 78, 84 and 90.}
\label{fig.new2}
\end{figure*}

\section{Conclusions}
\label{sectVI}
Transport processes between cold-core eddies and the Gulf Stream 
has been studied using a novel approach inspired on the Koopman \emph{observables}.
The flow field consists in a HYCOM database of the North Atlantic in which we advect particles at different isopycnal levels.
Given a basis set of functions, these observables are obtained by taking time-averages of such functions along
particle trajectories advected by the flow field.
As in other classical Lagrangian Coherent Structure (LCS) methods, the spatial scale of the flow can be modulated via 
appropriate integration time.

Once the Lagrangian statistics for the selected integration time span are obtained,
the resulting combined high-dimensional space of observables are partitioned using $k$-means.
The robustness of such partition was assessed by comparing the results for different
number of clusters and basis set definition and dimension.
Unlike classic LCS approaches based on local trajectory information on each isopycnal level,
the current analysis provides a natural representation of the three-dimensional
flow features that can be readily extended to account for temporal dependence by repeating
the analysis for different initial `particle launch' dates.

For the specific database used here, the results show that
the characteristic \emph{lobes} of detraining flow on the southern side of the jet eventually are stopped
when the mid-depth anticyclonic features growing in the bottom-up direction pinch-off this transport route.
At the same time and as the old-core eddy shape changes under the action of the Gulf Stream,
its vertical extent shrinks leading to volume reductions of $\approx$50\%
by the end of the 20-day period.
At deeper levels, the cold-core topology is observed to evolve from a arrowhead shape to a
dipole as the anticyclonic feature grows.

\section*{Acknowledgements}

This work was partially supported by  ONR Grants  N00014-11-1-0087 and N00014-14-10633.

\section{Appendix: Finite Time Ergodic (FiTEr) partition}
In this appendix we define the notion of Finite-Time Ergodic Partition (FiTEr partition)
and discuss some of its properties. 

We start with a point $\bx$ in a compact domain ${\cal D}$ of the flow,
and a sequence of continuous functions $\{f_i\}:\cD\rar \R$. Let 
\beq
\dot \bx=\bv(\bx,t)
\label{eq:vec}
\eeq
 be a vector field on $\cD$,
and its flow $S^t(\bx_0):\cD\rar S_{t_0}^t(\cD),$ where $S^t(\cD)$ is not necessarily equal to $\cD$ (i.e. 
trajectories can enter and exit $\cD$ during the time interval of observation). To each point $\bx\in \cD$ 
we associate a sequence of numbers $\{\tilde f_i^{t_0,T}(\bx)\}$ defined by 
\beq
\tilde f_i^{t_0,T}(\bx)=\frac{1}{T-t_0}\int_{t_0}^T f_i(S^t(\bx))dt=\frac{1}{T-t_0}\int_{t_0}^T U^t_{t_0}f_i(\bx))dt,
\label{eq:fiter}
\eeq
where $U^t_{t_0}$ is the finite-time Koopman operator associated with (\ref{eq:vec}).
For fixed $\bx$ and replacing $f_i$ with an arbitrary function $f$, Eq.\eqref{eq:fiter}
is a linear functional on the vector space of continuous functions on $\cD$ that, equipped with the Borel
$\sigma$-algebra, has a representation as a measure:
\beq
\tilde f^{t_0,T}(\bx)=\int_{\cD} f d\nu^{t_0,T}_\bx.
\label{eq:fitermeas}
\eeq
In Eq.\eqref{eq:fitermeas} $\nu^{t_0,T}_\bx$ is the measure that can be thought as a Dirac delta measure on
$\cD$ concentrated on the finite time part of the trajectory passing through $\bx$ at $t_0$
and ending at time $T$.
The family of measures $\nu^{t_0,T}_\bx$ characterizes the average
action in the time interval $[t_0,T]$ of the Koopman operator on observables,
$$
\frac{1}{T-t_0}\int_{t_0}^T U^t_{t_0}f_i(\bx))dt.
$$
Note the similitudes with the ergodic measures
(that are eigenfunctions of the Koopman operator adjoint, the Perron-Frobenius operator)
for vector fields in Eq.\eqref{eq:vec} that are independent of $t$ \cite{Mezic:1994}.

To get an understanding of the measure $\nu^{t_0,T}_\bx$, we would need to evaluate it on all
the continuous functions on $\cD$, which is unfeasible.
Instead, we prove that taking a set of functions $\{f_i^{t_0,T}(\bx)\}$
whose linear combinations are dense in $C(\cD)$ and evaluating the finite-time averages as defined in
\eqref{eq:fiter} leads to a unique determination of the measure $\nu^{t_0,T}_\bx$:
\begin{theorem}
Let $\nu^{t_0,T}_\bx$ be the measure determined by Eq.\eqref{eq:fitermeas} and $\tilde f_i^{t_0,T}(\bx)$
the finite-time averages defined in Eq.\eqref{eq:fiter}.
If linear combinations of continuous functions in $\{ f_i\}$ are dense in $C(\cD)$, then $\nu^{t_0,T}_\bx$
is uniquely determined by the set of values $\{\tilde f_i^{t_0,T}(\bx)\}$.
\end{theorem}
\begin{proof}
Assume not.
Then there is a measure $\mu\neq \nu^{t_0,T}_\bx$ s.t. $\tilde f_i^{t_0,T}(\bx)=\int_\cD f_i d\mu, \forall i$.
Consider an arbitrary function $f\in C(\cD)$.
Then there is a sequence of $g_i, i\in {\mathbb N}$ that are linear combinations of $f_i$ such that
$\lim_{i\rar\infty} g_i=f$.
Now, each $g_i$ is continuous, and convergence of $g_i$ to the continuous function $f$ on compact $\cD$
implies that they are bounded above by some continuous function.
By Lebesgue's Dominated Convergence, then,
$\tilde \lim_{i\rar\infty}g_i^{t_0,T}(\bx)=\lim_{i\rar\infty}\int_\cD g_i d\mu=\int_{cD}f d\mu$.
But the sequence $\tilde g_i^{t_0,T}(\bx)$ converges to $\tilde f^{t_0,T}(\bx)=\int_\cD f d \nu^{t_0,T}_\bx$
and thus we get a contradiction.
\end{proof}

The result above also has the following simple but important corollary:
\begin{corollary}
The family of measures $\nu^{t_0,T}_\bx$ does not depend on the chosen set of functions $\{ f_i\}$.
\end{corollary}
Note that, physically, the measure $\nu^{t_0,T}_\bx$ is just the {\it ``sojourn measure"} that
assigns to any open set $A$ in $\cD$ the amount of time that the trajectory starting at $\bx$
at time $t_0$ spends in $A$ in the period $[t_0,T]$.
This leads us to call $\nu^{t_0,T}_\bx$ the family of sojourn measures. In the limit $t\rar \infty$,
for steady, periodic or quasiperiodic velocity fields, sojourn measures converge to ergodic
measures (see \cite{Mezic:1994} for the proof of the ergodic partition theorem that uses the same
ideas as above, in the infinite-time limit, and additionally proves ergodic properties of the
family of measures thus obtained).

\begin{definition}
A FiTEr (Finite-Time-Ergodic) partition associated with a finite set of functions $\{f_i\}, i=1,...,n$ is the partition of
$\cD$ into joint level sets of $\tilde f_i^{t_0,T}(\bx), i=1,...n$.
\end{definition}

We remark on two consequences of the above result:
\begin{inparaenum}
\item the approaches based on integrals of positive functions \cite{manchoetal:2013} along a trajectory
can fail to uniquely identify underlying objects. Specifically, the notion of Lagrangian descriptors was
based on this, but the notion is neither new - as it was proposed much earlier,
in \cite{Mezic:1994,Poje1999} - nor complete, because it does not have the uniqueness property due to
choice of only positive functions and
\item attempts to identify invariant structures using averages such as pair
dispersion \cite{Haller:2015} also suffer from the same non-uniqueness problem.
\end{inparaenum}

Farazmand results in \cite{farazmand:2015} clarify the relationship between
averaging of scalars over Lagrangian trajectories and some of the geometrical approaches.
The above measure-theoretic approach shows that single scalar averages are not enough
and gives a comprehensive theory, as opposed to geometric approaches.
In some cases, even over finite time, sets of points are identified - e.g.\ for specific times,
corresponding to the period, invariant closed curves have the same finite-time averages for all their points.
But, in general, it is fruitful to embed the $n$ averages $\tilde f_i,i=1,...,n$ into
an $n$ dimensional space and cluster them, as done in this paper to identify zones with similar trajectory behavior.

Another connection of the this formalism with existing methods can be done as follows:
the notion of almost invariant sets has been advanced as a concept that captures fluid parcels traveling together,
initially Dellnitz \cite{Dellnitz1999, Dellnitz2002}.
The most recent application of the notion is that of finite-time coherent sets \cite{froyland2015}.
A metrically invariant set can be defined as a set $A$ that satisfies
$$
\int_A \left( \int \chi_A d\nu_\bx^{t_0,T} \right) d\mu=1.
$$
The above statement just means that the time average on interval $[t_0,T]$ of the indicator function
$\chi_A$ of $A$ is 1 on almost every trajectory starting in $A$ at time zero is 1.
Thus, {\it maximally metrically invariant sets} are those that satisfy
\beq
A_{maxinv}=\argmin_{A\in {\cal B}}|\int_A  \left(\int \chi_A d\nu_\bx^{t_0,T} \right)d\mu-1|,
\label{maxinv}
\eeq
the $\argmin$ exists, where ${\cal B}$ is the Borel $\sigma$-algebra on $\cD$.
Of course, we could ask for the lower bound of the expression
$\epsilon=|\int_A  \left(\int \chi_A d\nu_\bx^{t_0,T} \right) d\mu-1|$ and find
{\it almost maximally invariant sets}, parametrized by the value of $\epsilon$. 
The approach we advanced here can also be used to define finite-time maximally coherent sets via an extremum principle.
We say that a set $A\in {\cal B}$ is maximally coherent provided it satisfies
\beq
A_{maxcoh}=\argmin_{A\in {\cal B}}\frac{1}{2(\mu_A)^2}\int_{A\times A}  d(\nu_\bx^{t_0,T},\nu_\by^{t_0,T}) d\mu(\bx)d\mu(\by),
\label{maxcoh}
\eeq
where $d(\nu_\bx^{t_0,T},\nu_\by^{t_0,T})$ is some chosen distance on the space of
measures (for example the Wasserstein metric), and the factor of $2$ in the denominator is
due to the fact that $d(\nu_\bx^{t_0,T},\nu_\by^{t_0,T})=d(\nu_\by^{t_0,T},\nu_\bx^{t_0,T})$.
Once again, if $\argmin$ does not exist, we could ask for lower bounds, as done above for almost invariant sets.
This definition makes sense for a deterministic dynamical system,
without any artificial addition of stochastic dynamics as in \cite{froyland2015}.
It makes intuitive sense, since the idea is that trajectories at $\bx$ and $\by$ are compared based
on their statistical properties - in particular, their finite-time statistics. 
If those properties are close for a whole set of trajectories starting in $A$ then we proclaim that set coherent. 

One can develop a multi-scale approach this way, both in space and time, by defining scales
$\tau_s$ in time and $\tau_{\mathbf l}$ in space,
using a set of functions of time $\psi_{\tau_s}^j(t-c_s)$ and space $f^m(\bx-\bc^k_{\tau_{\mathbf l}})$,
where $j$ is the counter of functions at scale $\tau_s$ in time,
$m$ is the counter of functions at scale $\tau_{\mathbf l}$ in space,
$c_s$ are centerpoints of functions $\psi$ in the interval $t_0,T$ of observation
and $\bc^k_{\tau_{\mathbf l}}$ are centerpoints of functions $f$ defined on the domain $\cD$.
We form convolution integrals
\beq
\int_{t_0}^T \psi_{\tau_s}^j(t-c_s) f^m(\bx-\bc^k_{\tau_{\mathbf l}})=\tilde f_{\tau_{\mathbf l},m}^{{\tau_s},j}(\bx).
\eeq
if there is a countable set of scales in time and space, i.e.\,
a countable set of quantities out of which we can define a trajectory measure at $\bx$ as above.
In this paper, we chose $\psi$ to be functions that are constant on subintervals of
$[t_0,T]$ and $f$ as wavelet functions or harmonic functions in space.

%\bibliographystyle{unsrt}
%\bibliography{muri,koop,igorIM}

\end{document}